\def\year{2020}\relax
\documentclass[letterpaper]{article} 
  	\usepackage[hscale=0.7,scale=0.75]{geometry}
	\usepackage[numbers,sort&compress]{natbib}

\usepackage{ifluatex}
\ifluatex
\usepackage{luatex85}
\usepackage[OT1]{fontenc}
\fi

\usepackage{times}  
\usepackage{helvet} 
\usepackage{courier}  
\usepackage[hyphens]{url}  
\usepackage{graphicx} 
\usepackage{graphicx}  
\title{Structural Decompositions of Epistemic Logic Programs}

\author{
Markus Hecher\textsuperscript{\rm 1,2}$,$
Michael Morak\textsuperscript{\rm 3}$,$\and
Stefan Woltran\textsuperscript{\rm 1}\\
\textsuperscript{\rm 1}TU Wien, Vienna, Austria,\\
\textsuperscript{\rm 2}University of Potsdam, Potsdam, Germany\\
\textsuperscript{\rm 3}University of Klagenfurt, Klagenfurt, Austria\\
$\{\text{hecher}, \text{woltran}\}$@dbai.tuwien.ac.at, michael.morak@aau.at}

\usepackage{amsthm,amsmath,amssymb,wasysym,mathrsfs}

\newcommand{\nop}[1]{}

\newenvironment{changemargin}[2]{%
\list{}{\rightmargin#2\leftmargin#1
\parsep=0pt\topsep=0pt\partopsep=0pt}
\item[]}
{\endlist}

\newenvironment{indented}{\begin{changemargin}{1cm}{0cm}}{\end{changemargin}}

\newtheorem{theorem}{Theorem}
\newtheorem{corollary}[theorem]{Corollary}
\newtheorem{proposition}[theorem]{Proposition}
\newtheorem{lemma}[theorem]{Lemma}
\newtheorem{definition}[theorem]{Definition}
\newtheorem{example}[theorem]{Example}

{%
  \addtocounter{theorem}{-1}
  \endgroup
}%

\let\phi\varphi
\let\epsilon\varepsilon

\renewcommand{\models}{\vDash}

\newcommand{\calA}{\mathcal{A}}

\newcommand{\calC}{\mathcal{C}}

\newcommand{\calI}{\mathcal{I}}

\newcommand{\calK}{\mathcal{K}}
\newcommand{\calL}{\mathcal{L}}
\newcommand{\calM}{\mathcal{M}}

\newcommand{\calO}{\mathcal{O}}

\newcommand{\calR}{\mathcal{R}}
\newcommand{\calS}{\mathcal{S}}
\newcommand{\calT}{\mathcal{T}}

\newcommand{\scrP}{\mathscr{P}}

\newcommand{\algo}[1]{\ensuremath{\mathsf{#1}}}

\newcommand{\SIGMA}[2]{\ensuremath{\Sigma_{\mathit{#1}}^{\mathit{#2}}}}

\newcommand{\bigO}[1]{\ensuremath{{\mathcal O}(#1)}}

\newcommand{\tw}[1]{\mathit{tw}(#1)}

\newcommand{\constant}[1]{\mathit{#1}}

\newcommand{\variable}[1]{\mathit{#1}}
\newcommand{\variables}[1]{{\mathbf{#1}}}

\newcommand{\mods}[1]{\mathit{mods}(#1)}
\newcommand{\answersets}[1]{\mathit{AS}(#1)}

\newcommand{\cwvs}[1]{\mathit{cwv}(#1)}
\newcommand{\wvs}[1]{\mathit{wv}(#1)}

\newcommand{\varR}{\variable{R}}

\newcommand{\varX}{\variable{X}}

\newcommand{\varZ}{\variable{Z}}
\newcommand{\varsN}{\variables{N}}
\newcommand{\varsP}{\variables{P}}
\newcommand{\varsU}{\variables{U}}

\newcommand{\varsX}{\variables{X}}
\newcommand{\varsY}{\variables{Y}}

\newcommand{\relation}[1]{{\mathit{#1}}}

\newcommand{\fullatom}[2]{{\relation{#1}(#2)}}
\newcommand{\domof}[1]{\mathit{dom}(#1)}

\newcommand{\elitof}[1]{\mathit{atel}(#1)}
\newcommand{\atomof}[1]{\mathit{at}(#1)}

\newcommand{\eneg}{\mathbf{not}\,}

\newcommand{\body}[1]{{\mathit{B}(#1)}}

\newcommand{\pbody}[1]{{\mathit{B}^+(#1)}}
\newcommand{\head}[1]{{\mathit{H}(#1)}}

\DeclareMathOperator{\type}{type}
\newcommand{\intr}{\textsf{intr}}
\newcommand{\leaf}{\textsf{leaf}}
\newcommand{\rem}{\textsf{rem}}
\newcommand{\join}{\textsf{join}}
\newcommand{\tab}[1]{\ensuremath{\tau(#1)}}

\newcommand{\connvert}[1]{\mathit{conn}(#1)}

\newcommand{\tuplecolor}[1]{\textcolor{#1}}
\newcommand{\epistemiccolor}{orange!55!red}
\newcommand{\modelcolor}{blue!45!black}
\newcommand{\survivalcolor}{green!62!black}
\newcommand{\killcolor}{red!62!black}
\newcommand{\MAIR}[2]{\ensuremath{#1^+_{#2}}}%

\DeclareMathOperator{\EP}{\mathbb{E}}
\DeclareMathOperator{\GP}{\mathbb{P}}

\DeclareMathOperator{\updtCand}{updT}
\DeclareMathOperator{\intrCand}{intT}
\DeclareMathOperator{\updtCands}{intTs}

\DeclareMathOperator{\survivalSets}{succS}

\usepackage{multirow}
\usepackage{tikz}
\usetikzlibrary{positioning}
\usetikzlibrary{shapes}
\usetikzlibrary{calc}

\usepackage{colortbl}
\usetikzlibrary{positioning}

\tikzstyle{tdnode} = [draw,rounded corners,top color=vertexTopColor,bottom color=vertexBottomColor,minimum size=1.5em]
\tikzstyle{stdnode} = [tdnode, font=\scriptsize]
\tikzstyle{stdnodecompact} = [stdnode, inner sep = 1.5pt, outer sep = 0.1pt]
\tikzstyle{stdnodetable} = [stdnode, inner sep = 0.5pt, outer sep = 0]
\tikzstyle{stdnodenum} = [minimum size=1.5em, font=\scriptsize]
\tikzstyle{tdedge} = [-,draw,thick]
\tikzstyle{tdlabel} = [draw=none, rectangle, fill=none, inner sep=0pt, font=\scriptsize]
\colorlet{vertexTopColor}{white}
\colorlet{vertexBottomColor}{black!10}

\begin{document}

\maketitle

\begin{abstract}
Epistemic logic programs (ELPs) are a popular generalization of standard Answer
Set Programming (ASP) providing means for reasoning over answer sets within the
language.  This richer formalism comes at the price of higher computational
complexity reaching up to the fourth level of the polynomial hierarchy.
However, in contrast to standard ASP, dedicated investigations towards
tractability have not been undertaken yet.  In this paper, we give first results
in this direction and show that central ELP problems can be solved in linear
time for ELPs exhibiting structural properties in terms of bounded treewidth.
We also provide a full dynamic programming algorithm that adheres to these
bounds. Finally, we show that applying treewidth to a novel dependency
structure---given in terms of epistemic literals---allows to bound the number of
ASP solver calls in typical ELP solving procedures.
\end{abstract}

\section{Introduction}\label{sec:introduction}

Epistemic logic programs (ELPs) \cite{lpnmr:GelfondP91}, also referred to as the
language of Epistemic Specifications \cite{aaai:Gelfond91}, have received
renewed attention in the research community as of late. ELPs are an extension of
the language of Answer Set Programming (ASP) \cite{cacm:BrewkaET11,ki:SchaubW18}
with epistemic operators. \citeauthor{aaai:Gelfond91}
(\citeyear{aaai:Gelfond91}) introduced the operators $\mathbf{K}$ and
$\mathbf{M}$ in order to represent the concepts of \emph{known to be true} and
\emph{may be true}, and defined an initial semantics. Several improvements to
the semantics have since been proposed in the literature \cite{lpnmr:Gelfond11,%
birthday:Truszczynski11,logcom:KahlWBGZ15,ijcai:CerroHS15}.
\citeauthor{ai:ShenE16} (\citeyear{ai:ShenE16}) realized that these two
operators can be represented via a single negation-type operator that they
called \emph{epistemic negation}, denoted $\eneg$, and gave a new semantics
based on this operator. \citeauthor{iclp:Morak19} (\citeyear{iclp:Morak19})
proposed a novel characterization of the central construct of the ELP semantics:
the \emph{world view}. While a recent analysis \cite{lpnmr:CabalarFC19} has shown
that this semantics still does not eliminate all existing flaws, we will make
use of it in this paper, since no clear ``winner'' semantics has as of yet
emerged, and our approach should be equally applicable to other existing
semantics that have been proposed.

Evaluating ELPs is a computationally hard task. The central decision problem,
checking whether an ELP has a world view, is \SIGMA{3}{P}-complete, and problems
even higher on the polynomial hierarchy exist \cite{ai:ShenE16,iclp:Morak19}. In
order to deal with this high complexity efficiently, we propose to use a method
from the field of parameterized complexity, namely, investigate how the runtime
behaves when looking at different structural parameters of the problem.  For
standard ASP, this topic has received considerable interest,
\cite{tocl:LoncT03,ai:GottlobPW10,ai:FichteS15,ecai:BliemOW16,amai:FichteKW19,lpnmr:FichteH19}.
However, the parameterized complexity of epistemic ASP has remained largely
unexplored so far. From the ASP case, we see strong evidence that this will be
the case for ELPs as well. In this paper, we will investigate, in particular,
whether ELPs can be solved efficiently if their treewidth (i.e., a measure for
the tree-likeness of graphs) is bounded.

It turns out that this question can be answered in the affirmative: the main
decision problems become tractable. In practice, a dynamic programming algorithm
on tree decompositions can be used to exploit this directly, similarly to what
was successfully proposed for ASP and QBF solvers \cite{lpnmr:FichteHMW17,%
ijcai:BliemMMW17,cp:FichteHZ19,fi:CharwatW19}.
However, we also aim to investigate a more interesting angle. Many ELP solvers
today work by making (up to exponentially many) calls to an underlying ASP
solving system in order to check world view existence. Being able to find a bound
on the number of these ASP solver calls would be very useful. Using
so-called 
epistemic (primal) graphs
of ELPs that focus on epistemically negated
literals only, we can again employ treewidth to establish such bounds. This
novel use of structural decomposition intuitively works well in some interesting
cases including instances of the scholarship eligibility (SE)\footnote{Problem SE was a prime motivator for ELPs~\cite{aaai:Gelfond91}.} benchmark set, as provided with the system ``EP-ASP''~\cite{ijcai:SonLKL17}. 

Using the
epistemic primal graph representation, these instances naturally decompose into their
individual sub-problems, that is, one sub-instance of the SE problem for each
student within the original instance.

\medskip\noindent\textbf{Contributions.} Our contributions are summarized below:
\begin{itemize}
  \item We investigate the complexity of the ELP world view existence problem
    when parameterized by the treewidth of the ELP instance. We establish that
    this problem is fixed-parameter tractable in this setting, and, in fact, can
    be solved in linear time if the treewidth is bounded from above by a
    constant. The same holds for the even
    more complex problem  of	world view formula evaluation.
  \item Then, we propose a novel graph representation of ELPs, namely, the 
	  epistemic primal graph
	  and show how this can be exploited to bound the number of calls to
    an underlying ASP solver in a classical ELP solver setting. It turns out
    that the number of calls is bounded in case the epistemic primal graph has bounded
    treewidth.
  \item Finally, we provide a full dynamic programming algorithm that could be used in
    practice to directly exploit the tractability result above. We also show that
    %
    the 
    worst-case
    runtime of 
    this
    algoritm 
    %
    cannot be significantly improved 
    under reasonable complexity-theoretic assumptions.
\end{itemize}


\section{Preliminaries}\label{sec:preliminaries}

\smallskip\noindent\textbf{Answer Set Programming (ASP).} A \emph{ground logic program} with
nested negation (also called answer set program, ASP program, or, simply, logic
program) is a pair $\scrP = (\calA, \calR)$, where $\calA$ is a set of
propositional (i.e., ground) atoms and $\calR$ is a set of rules of the form
%
  %
  $a_1\vee \cdots \vee a_l \leftarrow a_{l+1}, \ldots, a_m, \neg \ell_1, \ldots,
  \neg \ell_n$,
  %
%
where the comma symbol stands for conjunction, $0 \leq l \leq m$, $0 \leq n$,
$a_i \in \calA$ for all $1 \leq i \leq m$, and each $\ell_i$ is a
\emph{literal}, that is, either an atom $a$ or its (default) negation $\neg a$
for any atom $a \in \calA$.\footnote{In this case, we say that it is a literal
\emph{over} $\calA$.} Note that, therefore, doubly negated atoms may occur. We
will sometimes refer to such rules as \emph{standard rules}.  Each rule $r \in
\calR$ consists of a \emph{head} $\head{r} = \{
a_1,\ldots,a_l \}$ and a \emph{body} $\body{r} = \{a_{l+1},\ldots,a_m, \neg
\ell_1, \ldots, \neg \ell_n \}$, and is alternatively denoted by
$\head{r}\leftarrow \body{r}$. The \emph{positive} body is given by $\pbody{r} =
\{ a_{l+1}, \ldots, a_m \}$.
%
Sometimes, we add a set of rules $\calR'$ to a logic program
$\scrP = (\calA, \calR)$. 
By some abuse of notation, let $\scrP \cup \calR'$ denote
the logic program $(\calA \cup \calA', \calR\cup \calR')$, where $\calA'$ is the set of atoms occurring in the rules of
$\calR'$.

An \emph{interpretation} $I$ (over $\calA$) is a set of atoms, that is, $I
\subseteq \calA$. A literal $\ell$ is true in an interpretation $I \subseteq
\calA$, denoted $I \models \ell$, iff $a \in I$ and $\ell = a$
, or $a \not\in I$ and $\ell = \neg a$; 
otherwise $\ell$ is false in $I$, denoted $I \not\models
\ell$. Finally, for some literal $\ell$, we define that $I \models \neg \ell$ if
$I \not\models \ell$. This notation naturally extends to sets of literals. An
interpretation $M$ is called a \emph{model} of $r$, denoted $M \models r$, if,
whenever $M \models \body{r}$, it holds that $M \models \head{r}$. We denote the
set of models of $r$ by $\mods{r}$; the models of a logic program $\scrP=
(\calA,\calR)$ are given by $\mods{\scrP} = \bigcap_{r \in \calR} \mods{r}$. We
also write $I \models r$ (resp.\ $I \models \scrP$) if $I \in \mods{r}$ (resp.\ $I
\in \mods{\scrP}$).

The GL-reduct 
of a 
logic program $\scrP = (\calA, \calR)$ with
respect to an interpretation $I$ is 
given by $\scrP^I=(\calA, \calR^I)$ 
with
$\calR^I = \{ \head{r} \leftarrow \pbody{r} \mid r \in \calR, \forall (\neg
\ell) \in \body{r} : I \not\models \ell \}$.

\begin{definition}[\cite{iclp:GelfondL88,ngc:GelfondL91,amai:LifschitzTT99}]
  \label{def:answerset}
  $M \subseteq \calA$ is an \emph{answer set} of a logic program $\scrP$ if (1) $M
  \in \mods{\scrP}$ and (2) there is no subset $M' \subset M$ such that $M' \in
  \mods{\scrP^M}$.
\end{definition}

The set of answer sets of a logic program $\scrP$ is denoted by
$\answersets{\scrP}$. The \emph{consistency problem} of ASP, that is, to decide
whether for a given logic program $\scrP$ it holds that $\answersets{\scrP} \neq
\emptyset$, is \SIGMA{P}{2}-complete~\cite{amai:EiterG95}, and remains so also
in the case where doubly negated atoms are allowed in rule
bodies~\cite{tplp:PearceTW09}.

\smallskip\noindent\textbf{Epistemic Logic Programs.} An \emph{epistemic literal} is a formula
$\eneg \ell$, where $\ell$ is a literal and $\eneg$ is the epistemic negation
operator. A \emph{ground epistemic logic program (ELP)} is a pair $\Pi = (\calA,
\calR)$, where $\calA$ is a set of 
atoms and $\calR$ is a set of
\emph{ELP rules}, which are implications of the form
%
  %
   $a_1\vee \cdots \vee a_k \leftarrow \ell_1, \ldots, \ell_m, \xi_1, \ldots,
   \xi_j, \neg \xi_{j + 1},$ $\ldots, \neg \xi_{n}$,
  %
%
where each $a_i$ is an atom from $\calA$, each $\ell_i$ is a literal over
$\calA$, and each $\xi_i$ is an \emph{epistemic literal} of the form $\eneg
\ell$, where $\ell$ is a literal over $\calA$. 
Similarly to logic programs, let $\head{r} = \{ a_1, \ldots, a_k \}$, and let
$\body{r} = \{ \ell_1, \ldots, \ell_m, \xi_1, \ldots, \xi_j, \neg \xi_{j+1},
\ldots, \neg \xi_{n} \}$. 
%
Further, $\atomof{r} \subseteq \calA$ denotes the set of atoms ocurring in
ELP rule $r$, and $\elitof{r} \subseteq \atomof{r}$ denotes the
set of atoms used in epistemic literals of $r$. These notions naturally
extend to sets of rules.

In order to define the semantics of an ELP, we will use the approach by
\citeauthor{iclp:Morak19} \cite{iclp:Morak19}, which follows the semantics
defined in \cite{ai:ShenE16}, but uses a different formal
representation. Note that, however, our results can be adapted to other
``reduct-based'' semantics, by changing the definition of the reduct
appropriately.
Given an ELP $\Pi = (\calA, \calR)$, a \emph{candidate world interpretation
(CWI)} $I$ for $\Pi$ is a consistent subset $I \subseteq \calL$, where $\calL$
is the set of all literals that can be built from atoms in $\calA$. Note that a
CWI $I$ naturally gives rise to a three-valued truth assignment to all the atoms
in $\calA$; hence, we will sometimes treat a CWI $I$ as a triple of disjoint
sets $\langle I^P, I^N, I^U \rangle$, where $I^P = \{ a \mid a \in I \cap \calA \}$, $I^N =
\{ a \mid \neg a \in I \}$ and $I^U = (\calA \setminus I^P) \setminus I^N$, with
the intended meaning that atoms in $I^P$, $I^N$, and $I^U$ are ``always true,'' ``always false,'' and ``unknown,'' respectively.

With the above definition in mind, we now define when a CWI is compatible with a
given set of interpretations.

\begin{definition}\label{def:compat}
  Let $\calI$ be a set of interpretations over a set of atoms $\calA$. Then, a
  CWI $I$ is \emph{compatible} with $\calI$ iff we have:
  \begin{enumerate}
    \item\label{def:compat:1} $\calI \neq \emptyset$;
    \item\label{def:compat:2} for each atom $a \in I^P$, it holds that for each
		  $J \in \calI$, $a \in J$;
    \item\label{def:compat:3} for each atom $a \in I^N$, we have for each
	    	$J \in \calI$, $a \not\in J$;
    \item\label{def:compat:4} for each atom $a \in I^U$, it holds that there are
	    $J, J' \in \calI$, such that $a \in J$, but $a \not\in J'$.
  \end{enumerate}
\end{definition}

The \emph{epistemic reduct} \cite{ai:ShenE16,iclp:Morak19} of program $\Pi =
(\calA, \calR)$ w.r.t.\ a CWI $I$, denoted $\Pi^I$, is defined as 
$\Pi^I = (\calA,  \{ r^I \mid r \in \calR \})$ where $r^I$
denotes rule $r$ where each epistemic literal $\eneg \ell$ is replaced by
$\neg \ell$ if $\ell \in I$, and by $\top$ otherwise.  Note that, hence, $\Pi^I$
is a plain logic program with all occurrences of epistemic negation
removed.\footnote{In fact, $\Pi^I$ may contain triple-negated atoms
$\neg\neg\neg a$.  
According to 
\cite{amai:LifschitzTT99}, such formulas are equivalent to simple negated atoms
$\neg a$, and we treat them as such.}

Now, a CWI $I$ is a \emph{candidate world view (CWV)} of $\Pi$ iff 
$I$ is compatible with
the set~$\answersets{\Pi^I}$ of
answer sets. The set of CWVs of an
ELP $\Pi$ is denoted $\cwvs{\Pi}$. Following the principle of knowledge
minimization, furthermore $I$ is a \emph{world view (WV)} iff it is a CWV and
there is no proper subset $J \subset I$ such that $J \in \cwvs{\Pi}$. The set of
WVs of an ELP $\Pi$ is denoted $\wvs{\Pi}$.

One of the main reasoning tasks regarding ELPs is the \emph{world view existence
problem}, that is, given an ELP $\Pi$, decide whether $\wvs{\Pi} \neq \emptyset$
(or, equivalently, whether $\cwvs{\Pi} \neq \emptyset$). This problem is known
to be \SIGMA{3}{P}-complete \cite{ai:ShenE16,iclp:Morak19}. Another interesting
reasoning task is deciding, given an ELP $\Pi = (\calA, \calR)$ and an
arbitrary propositional formula $\phi$ over $\calA$, whether $\phi$ holds in
some WV, that is, whether there exists $W \in \wvs{\Pi}$ such that $W \models
\phi$. This \emph{formula evaluation problem} is even harder, namely
\SIGMA{4}{P}-complete \cite{ai:ShenE16}.

\smallskip\noindent\textbf{Tree Decompositions and Treewidth.} %
We assume that graphs are undirected, simple, and free of self-loops. Let $G =
(V, E)$ be a graph, $T$ a rooted tree, and $\chi$ a labeling function that maps
every node $t$ of $T$ to a subset $\chi(t) \subseteq V$ called the \emph{bag} of $t$. The pair
$\calT = (T, \chi)$ is called a \emph{tree decomposition (TD)} \cite{jct:RobertsonS84} of $G$
iff (i) for each $v \in V$, there exists a $t$ in $T$, such that $v \in
\chi(t)$; (ii) for each $\{v, w\} \in E$, there exists $t$ in $T$, such that
$\{v, w\} \subseteq \chi(t)$; and (iii) for each $r, s, t$ of $T$, such that $s$
lies on the unique path from $r$ to $t$, we have $\chi(r) \cap \chi(t) \subseteq
\chi(s)$. 
For a node~$t$ of~$T$, we say that $\type(t)$ is $\leaf$ if $t$ has
no children and~$\chi(t)=\emptyset$; $\join$ if $t$ has children~$t'$ and $t''$ with
$t'\neq t''$ and $\chi(t) = \chi(t') = \chi(t'')$; $\intr$
(``introduce'') if $t$ has a single child~$t'$,
$\chi(t') \subseteq \chi(t)$ and $|\chi(t)| = |\chi(t')| + 1$; $\rem$
(``removal'') if $t$ has a single child~$t'$,
$\chi(t') \supseteq \chi(t)$ and $|\chi(t')| = |\chi(t)| + 1$. 
If for
every node $t\in T$, %
$\type(t) \in \{ \leaf, \join, \intr, \rem\}$, 
then $(T,\chi)$ is called \emph{nice}.
%
%
The \emph{width} of a TD is defined as the cardinality of its largest bag minus
one. The \emph{treewidth} of a graph $G$, denoted by $\tw{G}$, is the minimum
width over all TDs of $G$.  Note that if $G$ is a tree, then $\tw{G} = 1$.


\smallskip\noindent\textbf{Monadic Second Order Logic and Courcelle's Theorem.}
Monadic Second Order logic (MSO) extends First Order logic (FO) with set
variables that range over sets of domain elements. Atomic MSO-formulas over a
signature $\sigma$ are either (1) atoms over some predicate in $\sigma$; (2)
equality atoms; or (3) atoms of the form $x \in S$, where $x$ is a FO variable,
and $S$ is a set variable. MSO-formulas are closed under FO operators. It is
convenient to use symbols like $\not\in$, $\subseteq$, $\subset$, $\cap$, or
$\cup$, with the obvious meanings as abbreviations for the corresponding MSO
(sub-)formulas. A $\sigma$-structure $\mathfrak{A}$ is a set of atoms over
predicates in $\sigma$. Let $\domof{\mathfrak{A}}$ denote its domain.

In order to exploit the structural information, we need to define how logical
structures can be represented as graphs, and how their treewidth is then defined:
Given a structure $\mathfrak{A}$ over some logical signature $\sigma$ of arity
at most two (sufficent for us), 
we say
that the treewidth of $\mathfrak{A}$ equals $\tw{G_\mathfrak{A}}$, where
$G_\mathfrak{A} = (V, E)$ is a graph with $V = \domof{\mathfrak{A}}$
and edge $\{a, b\} \in E$ iff $r(a, b) \in \mathfrak{A}$, where $r$ is some
relation in $\sigma$.

MSO formulas over structures of bounded treewidth are important in the context
of parameterized complexity in order to establish running time bounds, as the
following landmark theorem by \citeauthor{iandc:Courcelle90} shows:

\begin{theorem}[\cite{iandc:Courcelle90}]\label{thm:courcelle}
  Let $\phi$ be a fixed MSO formula over signature $\sigma$ and let
  $\mathfrak{A}$ be a $\sigma$-structure with $\tw{\mathfrak{A}} \leqslant k$,
  for some integer $k$.  Then, evaluating $\phi$ over $\mathfrak{A}$ can be done
  in time $O(f(k) \cdot |\mathfrak{A}|)$, for some function $f$ not depending on
  $|\mathfrak{A}|$.
\end{theorem}

Problems with a parameter $k$ that can be solved in time $O(f(k) \cdot n^c)$, where~$c$ is a constant and $f$ only depends on $k$, are called
\emph{fixed-parameter tractable (FPT)} \cite{mcs:DowneyF99}.

\section{An MSO Encoding for ELPs}\label{sec:msoencoding}

The main objective in this section is to investigate how the semantics of ELPs
can be encoded in terms of an MSO formula and thereby investigate, from a
theoretical perspective, the time complexity of evaluating ELPs, specifically
looking at tree-like instances. 

Now, our goal will be to offer a fixed MSO encoding to exploit
Theorem~\ref{thm:courcelle}, in the spirit of \cite{ai:GottlobPW10}, which is
able to solve the world view existence problem for
an ELP by evaluating it over a suitable logical structure representing the ELP.
In order to begin the construction of this, we first need to fix the signature
over which our MSO encoding will be expressed. To this end, let signature\smallskip\\
$\sigma = \{ \relation{atom}, \relation{rule}, \relation{h}, \relation{b},
\relation{b^\neg}, \relation{b^\eneg}, \relation{b^{\eneg \neg}},
\relation{b^{\neg \eneg}}, \relation{b^{\neg \eneg \neg}} \}, $\smallskip\\ 
where $\fullatom{atom}{a}$ and $\fullatom{rule}{r}$ represent the fact that domain
elements $a$ and $r$ are an atom and a rule, respectively; where
$\fullatom{h}{a, r}$ represents that atom $a$ appears in the head of rule $r$;
and where $b^\square(a, r)$, with $\square \in \{ \epsilon, \neg, \eneg, \neg
\eneg,$ $\eneg \neg, \neg \eneg \neg \}$ and $\epsilon$ the empty word, represents
that fact that the sub-formula $\square a$, for atom $a$, appears in the
body of rule $r$. Next we construct the~encoding. 

\begin{lemma}\label{lem:msoencoding}
  %
  Consider the signature $\sigma$ above. WV existence can be
  expressed by means of a fixed MSO formula over $\sigma$.
\end{lemma}
\begin{proof}
  Recall that, in order to check the existence of a WV, it suffices to check the
  existence of a CWV (since WVs are simply subset-minimal CWVs). We will
  construct an MSO formula $\fullatom{cwv}{\varsP, \varsN, \varsU}$ with the
  intended meaning that it evaluates to true iff the input set variables
  $\varsP$, $\varsN$, and $\varsU$ represent a CWV $W$ with $W^P = \varsP$, $W^N
  = \varsN$, and $W^U = \varsU$. To this end, our formula will be of the
  following form:\smallskip\\
	$
   \fullatom{cwv}{\varsP, \varsN, \varsU} \equiv
	 \fullatom{cwi}{\varsP, \varsN, \varsU} \wedge \bigwedge_{i=1}^4 \fullatom{chk_i}{\varsP, \varsN, \varsU} 
	 $\smallskip\\
  where $\relation{cwi}$ ensures that $\varsP$, $\varsN$, and $\varsU$ indeed
  encode a valid CWI (i.e., a three-partition of the set of atoms stored in
  $\relation{atom}$), and $\relation{chk_i}$ verifies that Condition~$i$ of
  Definition~\ref{def:compat} holds. We will now give the construction of
  these checks.

  First, the check for a valid CWI is expressed as follows:\medskip\\
  %
    %
    $\fullatom{cwi}{\varsP, \varsN, \varsU} \equiv \forall \varX
    \left(\fullatom{atom}{\varX} \Leftrightarrow \varX \in \varsP \cup \varsN
    \cup \varsU\right) \wedge$ \\
    $\neg \exists \varX \left((\varX \in \varsP \cap \varsN) \vee (\varX \in
    \varsN \cap \varsU) \vee (\varX \in \varsP \cap \varsU) \right)$\medskip\\
    %
  %
  The four remaining checks have a similar structure:
  \noindent\begin{align*}
    \fullatom{chk_{\ref{def:compat:1}}}{\varsP, \varsN, \varsU} \equiv\ & \exists
    \varsX\, \fullatom{as}{\varsX, \varsP, \varsN, \varsU};\\[.1ex]
    %
    \fullatom{chk_{\ref{def:compat:2}}}{\varsP, \varsN, \varsU} \equiv\ & \forall
    \varX \left( \varX \in \varsP \Rightarrow \right.\\[-.25em]
    & \left. \forall \varsX \left( \fullatom{as}{\varsX, \varsP, \varsN, \varsU}
    \Rightarrow \varX \in \varsX \right) \right);\end{align*}
    \noindent\begin{align*}
    %
    \fullatom{chk_{\ref{def:compat:3}}}{\varsP, \varsN, \varsU} \equiv\ & \forall
    \varX \left( \varX \in \varsN \Rightarrow \right.\\[-.25em]
    & \left.\forall \varsX \left( \fullatom{as}{\varsX, \varsP, \varsN, \varsU}
    \Rightarrow \varX \not\in \varsX \right) \right);\\[.1ex]
    %
    \fullatom{chk_{\ref{def:compat:4}}}{\varsP, \varsN, \varsU} \equiv\ & \forall
    \varX \left( \varX \in \varsU \Rightarrow \right.\\[-.25em]
    & \left. \left( \exists \varsX \left( \fullatom{as}{\varsX, \varsP, \varsN,
    \varsU} \wedge \varX \in \varsX \right) \wedge \right. \right.\\[-.25em]
    & \left. \left. \exists \varsX \left( \fullatom{as}{\varsX, \varsP, \varsN,
    \varsU} \wedge \varX \not\in \varsX \right) \right) \right).
  \end{align*}
  The four checks encode precisely the conditions of
  Definition~\ref{def:compat}, where $\fullatom{as}{\varsX, \varsP, \varsN,
  \varsU}$ is a sub-formula, to be defined below, that expresses that $\varsX$
  is an answer set of the epistemic reduct w.r.t.\ the CWI represented by the
  sets $\varsP$, $\varsN$, and $\varsU$. For example,
  $\relation{chk_{\ref{def:compat:3}}}$ encodes that for each atom $\varX$ that is
  set to ``always false'' in the CWI (i.e., $\varX \in \varsN$), it must hold
  that for every stable model $\varsX$ of the epistemic reduct, $\varX$ must not
  be true in that stable model (i.e., $\varX \not\in \varsX$).

  It now remains to define the sub-formula for checking answer sets. This
  construction is based on the one presented in
  \cite[Theorem~3.5]{ai:GottlobPW10}, but adapted to take the computation of the
  epistemic reduct into account. Firstly, a set of atoms $M$ is an answer set if
  it is a model and no proper subset of $M$ is a model of the GL-reduct w.r.t.\
  $M$. This is expressed as follows (note that any model $M$ is also a model of
  its GL-reduct):\medskip\\
  %
    %
    \makebox[1em]{}$\fullatom{as}{\varsX, \varsP, \varsN, \varsU} \equiv \fullatom{gl}{\varsX,
    \varsX, \varsP, \varsN, \varsU}$\\
    \makebox[5em]{}$\wedge \forall \varsY \left( \varsY \subset \varsX \Rightarrow \neg
    \fullatom{gl}{\varsX, \varsY, \varsP, \varsN, \varsU} \right).$\smallskip\\
    %
  %
  The intuitive meaning of $\fullatom{gl}{\varsX, \varsY, \varsP, \varsN,
  \varsU}$ is that it should hold iff $\varsY$ is a model of the GL-reduct
  w.r.t.\ $\varsX$ of the epistemic reduct w.r.t.\ the CWI represented by
  $\varsP$, $\varsN$, and $\varsU$:
  \begin{multline*}
    \fullatom{gl}{\varsX, \varsY, \varsP, \varsN, \varsU} \equiv \forall
    \varR \left( \fullatom{rule}{\varR} \Rightarrow \exists \varZ
    \left(\right.\right.\\[-.25em]
    %
    \left.\left.(\fullatom{h}{\varZ, \varR} \wedge \varZ \in \varsY)\right.\right.\\[-.25em]
    %
    \left.\left.\vee (\fullatom{b}{\varZ, \varR} \wedge \varZ \not\in \varsY) \vee
    (\fullatom{b^\neg}{\varZ, \varR} \wedge \varZ \in \varsX))\right.\right.\\[-.25em]
    %
    \left.\left.\vee (\fullatom{b^\eneg}{\varZ, \varR} \wedge \varZ \in \varsP \wedge \varZ
    \in \varsX)\right.\right.\\[-.25em]
    %
    \left.\left.\vee (\fullatom{b^{\eneg\neg}}{\varZ, \varR} \wedge \varZ \in \varsN \wedge
    \varZ \not\in \varsX)\right.\right.\\[-.25em]
    %
    \left.\left.\vee (\fullatom{b^{\neg\eneg}}{\varZ, \varR} \wedge \left((\varZ \in \varsN
    \cup \varsU) \vee (\varZ \in \varsP \wedge \varZ \not\in \varsX))\right)\right.\right.\\[-.25em]
    %
    \hspace{-.5em}\left.\left. \vee (\fullatom{b^{\neg\eneg\neg}}{\varZ, \varR} \wedge \left((\varZ
    \in \varsP \cup \varsU) \vee (\varZ \in \varsN \wedge \varZ \in
    \varsX))\right)\right)\right)
  \end{multline*}
  Note that the definition of the
  $\relation{gl}$-relation is such that it precisely mirrors the definition of
  both the epistemic reduct and the GL-reduct. It amounts to checking that every
  rule in the GL-reduct is satisfied, and amounts to a large case distinction,
  dealing with all seven cases of how atoms can appear in a rule (that is,
  either in the head, or nested under six combinations of default and epistemic
  negation in the body). For example, in line three, the first disjunct says
  that rule $\varR$ is satisfied if there is an atom $\varZ$ in the positive
  body, but this atom is not present in the reduct model $\varsY$ (satisfying
  rule $\varR$ by not satisfying the body). Line four treats the case of an
  epistemically negated atom $\varZ$ in the body. Such a rule is satisfied iff
  $\varZ$ is set to ``always true'' in the CWI (since otherwise the epistemic
  literal is replaced by $\top$ in the epistemic reduct, and the rule cannot be
  satisfied solely by this body element in this case), and $\varZ$ is false in
  the original model $\varsX$ (since in the epistemic reduct, the epistemic
  negation turns into default negation in this case, and default-negated atoms
  are evaluated over the original model $\varsX$). 

  This completes our MSO encoding. Correctness follows by construction, as
  explained above. In order to solve the WV existence problem via this encoding,
  we simply have to quantify the relevant set variables:\smallskip\\ 
  \makebox[.1\textwidth]{}$\phi = \exists
  \varsP \exists \varsN \exists \varsU\, \fullatom{cwv}{\varsP, \varsN, \varsU}.$\smallskip\\
   Evaluating this formula over a $\sigma$-structure $\mathfrak{P}$ that
  represents an ELP $\Pi$, we get that $\Pi$ has a CWV iff $\mathfrak{P} \models
  \phi$.
\end{proof}

With the above reduction in mind, lets take a closer look at what worst-case
solving time guarantees we can give for solving ELPs, in particular w.r.t.\
structural properties. Let $\Pi = (\calA, \calR)$ be an ELP, and let
$\mathfrak{P}$ be the $\sigma$-structure that represents it.  Recall that
$\tw{\mathfrak{P}} = \tw{G_\mathfrak{P}}$. In our case, $G_\mathfrak{P}$
coincides with the so-called \emph{incidence graph} of the ELP $\Pi$, a graph
representation that is well-known and studied in the literature for a wide range
of logic-based formalisms, and, in particular, for ASP
\cite{ijcai:JaklPW09,lpnmr:FichteHMW17}. The incidence graph of an ELP $\Pi$ is
the graph $G = (V, E)$ with $V = \calA \cup \calR$ and $\{a, r\} \in E$ iff atom
$a \in \calA$ occurs in rule $r \in \calR$ in $\Pi$. It is not difficult to
verify that the $\sigma$-structure $\mathfrak{P}$, when represented as a graph,
mirrors the incidence graph of $\Pi$ precisely. From this correspondence,
Theorem~\ref{thm:courcelle}, and Lemma~\ref{lem:msoencoding}, we thus
obtain the following, fundamental parameterized complexity result:

\begin{theorem}\label{thm:elpfpt}
  Let $\Pi$ be an ELP, let $G$ be its incidence graph, and let $\tw{G} \leqslant
  k$, for some integer $k$. Then, checking whether $\cwvs{\Pi} \neq \emptyset$
  can be done in time $O(f(k) \cdot |\Pi|)$, for some function $f$ that does not
  depend on $|\Pi|$.
\end{theorem}

Using a simple extension of the MSO construction in the proof of
Lemma~\ref{lem:msoencoding}, we can state a similar result for the 
formula evaluation problem.
The 
MSO
formula\medskip\\ 
%
  %
  $\exists \varsP \exists \varsN \exists \varsU\, \fullatom{cwv}{\varsP, \varsN,
  \varsU} \wedge \fullatom{entails}{\varsP, \varsN, \varsU, \phi}$ \\
  $\wedge \neg \left( \exists \varsP' \exists \varsN' \exists \varsU'\, (\varsP'
  \subset \varsP \vee \varsN' \subset \varsN) \wedge \fullatom{cwv}{\varsP',
  \varsN', \varsU'} \right)$\smallskip\\
  %
%
checks whether there is at least one WV that satisfies formula $\phi$, where the
atom $\fullatom{entails}{\varsP, \varsN, \varsU, \phi}$ encodes the check that
the WV represented by $\varsP$, $\varsN$, and $\varsU$ cautiously entails
formula $\phi$, a straightforward model-checking construction left to the
interested reader.  We obtain the following by the same argument as the one for
Theorem~\ref{thm:elpfpt}:

\begin{theorem}\label{thm:elpentailmentfpt}
  Let $\Pi$, $G$, and $k$ be as in Theorem~\ref{thm:elpfpt}, and let $\phi$ be a
  propositional formula. Then, checking whether $\Pi$ has a WV that cautiously
  entails $\phi$ can be done in time $O(f(k) \cdot |\Pi|)$, for some function
  $f$ that does not depend on $|\Pi|$.
\end{theorem}

From the above theorems we immediately obtain the fact that ELPs of bounded
treewidth can be solved in linear time in the size of the ELP. 
We will
investigate how to exploit this result and pinpoint function $f(k)$ in the
next two sections.

\section{Bounding Calls to Standard ASP Solvers}\label{sec:torso}

Before providing a concrete algorithm for the FPT result in
Theorem~\ref{thm:elpfpt} we will investigate a more abstract approach. Many ELP
solvers today make use of standard ASP solvers to check the compatibility of a
CWI with the set of answer sets of its epistemic reduct. However, the number of
calls to such an ASP solver can be at worst exponential. In this section, we
will propose an algorithm that makes use of the structural relationships between
the epistemic literals in an ELP in order to control the number of ASP solver
calls needed and give finer-grained worst-case bounds on this number. In the next
section, we will then extend these concepts to a full dynamic programming
algorithm that exploits the result in Theorem~\ref{thm:elpfpt}. 

We first need to define the structural relationship between
atoms occurring in epistemic literals in an ELP. To this end, let $\Pi = (\calA,
\calR)$ be an ELP.  Then, the \emph{primal graph}~$\GP_\Pi = (V, E)$ of~$\Pi$ is
a graph with $V = \calA$ and $\{a,b\} \in E$ iff atoms $a$ and $b$ with $a\neq b$ appear
together in a rule in $\calR$, that is, iff $\exists r \in \calR : \{ a, b \}
\subseteq \atomof{r}$. Two vertices~$a,b$ in the primal graph are
\emph{non-epistemically connected} iff there is a path $\langle a, v_1, \ldots,
v_n, b \rangle$ with $n \geq 0$ in $\GP_\Pi$, such that each vertex $v_i$, $1
\leq i \leq n$, belongs to the set $\calA \setminus \elitof{\Pi}$. Now, the
\emph{epistemic primal graph}~$\EP_\Pi = (V, E)$ of~$\Pi$ is a graph with the
vertex set $V = \elitof{\Pi}$ being the set of atoms appearing in epistemic
literals in $\Pi$, and an edge $\{a, b\} \in E$ iff $a \neq b$ and vertices $a, b$
are non-epistemically connected in $\GP_\Pi$. Intuitively, two atoms from
$\elitof{\Pi}$ form an edge in $\EP_\Pi$ iff they are connected in $\GP_\Pi$ via
atoms that \emph{do not} appear in epistemic literals. The concept of epistemic
primal graph is inspired by the notion of the \emph{torso
graph}~\cite{stacs:GanianRS17}, which is used in parameterized complexity to
decompose certain abstraction 
graphs.

\begin{example}\label{ex:running2}
\label{ex:running1}
	Consider the classic \emph{scholarship eligibility problem}  encoding,
	first investigated by
	\citeauthor{aaai:Gelfond91}
	\cite{aaai:Gelfond91}:
	\smallskip\\
  %
    %
    $\fullatom{eligible}{\varX} \gets \fullatom{highGPA}{\varX}$\\ 
    $\fullatom{ineligible}{\varX} \gets \fullatom{lowGPA}{\varX}$\\
    $\bot \gets \fullatom{eligible}{\varX}, \fullatom{ineligible}{\varX}$\\
    $\fullatom{interview}{\varX} \gets \eneg \fullatom{eligible}{\varX},
    \eneg \fullatom{ineligible}{\varX}$.\smallskip
    %

  Now, assume the above abstract (non-ground) program is instantiated with two
  students (assume that it is copied twice and $\constant{mike}$ and
  $\constant{mark}$ are substituted for $\varX$), and that we add the following rules,
  resulting in ELP~$\Pi$:\smallskip\\
  %
    %
    $\fullatom{lowGPA}{\constant{mike}} \vee
    \fullatom{highGPA}{\constant{mike}}$\\
    $\fullatom{lowGPA}{\constant{mark}} \vee \fullatom{highGPA}{\constant{mark}}$\smallskip\\
    %
  %
%
Epistemic primal graph~$\EP_\Pi$ contains only four nodes:
  $\fullatom{eligible}{\constant{mike}}$,
  $\fullatom{ineligible}{\constant{mike}}$, and the same two for
  $\constant{mark}$. Further, $\EP_\Pi$ does not have any edges
  except an edge between the two $\constant{mike}$-atoms, and the same for
  $\constant{mark}$. Since $\EP_\Pi$ forms a forest, 
  the treewidth of $\EP_\Pi$ is $1$. 
\end{example}

While the epistemic primal graph does not directly provide new
complexity results, it will allow us to give firm guarantees on the number of
ASP solver calls needed. As a side-effect, this algorithm is
conceptually simpler than the one of the next section, but prepares ideas for
later.

Algorithms that solve problems of bounded treewidth typically proceed by
\emph{dynamic programming (DP)}, bottom-up, along a TD where at each node~$t$ of the
TD information is gathered~\cite{jal:BodlaenderK96} in a table~$\tab{t}$.
A \emph{table}~$\tab{t}$ is a set of rows, where a \emph{row}~in $\tab{t}$ is a
fixed-length sequence of elements.
Tables are derived by an algorithm executed in each bag, called
\emph{bag algorithm}, which determine the actual content and meaning of the rows.  
Then, the DP approach~$\algo{DP}_\algo{B}$ for an epistemic logic program~$\Pi$ and
a given bag algorithm~$\algo{B}$ performs the following steps:
\begin{enumerate}
	\item Construct graph representation~$G$ of~$\Pi$ that is used by~$\algo{B}$.
	\item Heuristically compute a (nice) TD~$\calT=(T,\chi)$~of~$G$.
	\item\label{step:dp} Execute~$\algo{B}$ for every node~$t$ in TD $\calT$
	  in post-order.
  As input, $\algo{B}$ takes 
  a node~$t$, a bag~$\chi(t)$, a \emph{solving program} (depending on~$\chi(t)$
  and~$\algo{B}$), which is the part of~$\Pi$ currently visible in $t$, and the tables 
  computed at children of~$t$. 
  Bag algorithm~$\algo{B}$ outputs a
  table~$\tab{t}$. 
	\item Print the result by interpreting the table for root~$n$ of~$T$.
\end{enumerate}

%
%
%
%
Next, we define a bag algorithm~\algo{EPRIM} for the epistemic primal graph
representation of~$\Pi$.  To this end, let~$\Pi = (\calA, \calR)$ be the
given input epistemic program, $\calT = (T, \chi)$ be a nice TD of~$\EP_\Pi$, $t$ be
a node of~$\calT$, and~$\prec$ be any arbitrary total ordering among the nodes
in~$\calT$.
To ease notation, for some set~$X\subseteq\elitof{\calR}$, let $\connvert{X}$ be
the set of vertices (i.e., atoms) from $\GP_\Pi$ that lie on a path that
non-epistemically connects any two vertices $a$ and $b$ in $X$.\footnote{Note
that we may have that $a = b$, and hence, $\connvert{\{ a \}}$ contains all
those vertices from $\calA \setminus \elitof{\Pi}$ connected to atom $a$.}
We now define the \emph{induced bag rules} for node $t$ of
$\calT$, denoted by $\calR_t^{\EP}$, as follows. For every rule $r \in \calR$,
$r$ is \emph{compatible} with node $t$ of $\calT$ iff  (a) $\atomof{r} \cap\connvert{\elitof{\calR}} \subseteq \connvert{\chi(t)}$, and (b)
$\chi(t)$ is subset-maximal among all nodes of $\calT$.
Now, $r \in
\calR_t^{\EP}$ iff $t$ is the $\prec$-minimal node among all nodes~$t'$ in $\calT$ with~$\type(t')=\intr$ compatible with $r$. 
The \emph{induced bag program}
for node $t$ is the ELP $\Pi_t^{\EP} = (\atomof{\calR_t^{\EP}},
\calR_t^{\EP})$.
%
%

Observe that any set of vertices that form a clique within $\EP_\Pi$ will appear
together in some node $t$ of $\calT$. Note that for each node $t$ of $\calT$
that has not a non-subset-maximal bag, or has one, but is not $\prec$-minimal
for any compatible~$r\in\calR$, the
induced bag program is empty. Therefore, we have that for each rule~$r \in
\calR$ there is exactly one node $t$ of $\calT$ where~$r \in \calR_t^{\EP}$,
and, even more stringent, that each atom $a \in \atomof{r} \setminus \elitof{r}$
appears only in the induced bag program of $t$, but not in any other node.
The bag algorithm~\algo{EPRIM} uses the induced bag program as its solving
program, and, following the argument above, can check all rules containing atoms
from $\calA \setminus \elitof{\Pi}$ in a single node. Hence, during its
traversal of the tree decomposition, it does not need to store anything about
these atoms.
Instead, every row computed as part of a table by \algo{EPRIM} for a node~$t$,
called an \emph{epistemic row}, is of the form~$\langle I \rangle$, where~$I
\subseteq 2^{\{a, \neg a \mid a \in \chi(t) \}}$ is a \emph{partial CWI} (that
is, a CWI restricted to and defined w.r.t.\ $\chi(t)$)\footnote{This also means
that $J^P$, $J^N$, and $J^U$ are defined w.r.t.\ $\chi(t)$.}. For easy
identification, the CWI part of the rows in our algorithms will always be
printed using an orange color.
%

Listing~\ref{alg:eprim} presents algorithm~\algo{EPRIM}. For the ease of
presentation, it deals with nice TDs only, 
but can be generalized to arbitrary TDs, requiring
a more involved case distinction.
Intuitively, since for each leaf node~$t$ we have~$\chi(t)=\emptyset$,
bag algorithm~\algo{EPRIM} ensures in Line~\ref{row:leaf} that~$\tab{t}$ consists only of the empty epistemic row.
Then, when an atom~$a$ appears in bag~$\chi(t)$ for a node~$t$, 
but does not occur in child bags,
CWI~$J$, with either~$a\in J^P$, $a\in J^N$, or~$a\in J^U$, is computed in Line~\ref{line:introduce}.
Further, if the solving program with rules ~$\calR_t^{\EP}$ is not empty, i.e., $t$ is the unique node responsible for evaluating all the rules in~$\calR_t^{\EP}$,
the four conditions of Definition~\ref{def:compat} are checked in Line~\ref{line:compat}. 
Note that these checks can be performed by calling a black-box ASP solvers a
limited number of times for
each row in~$t$:

{\center\includegraphics[width=.6\textwidth]{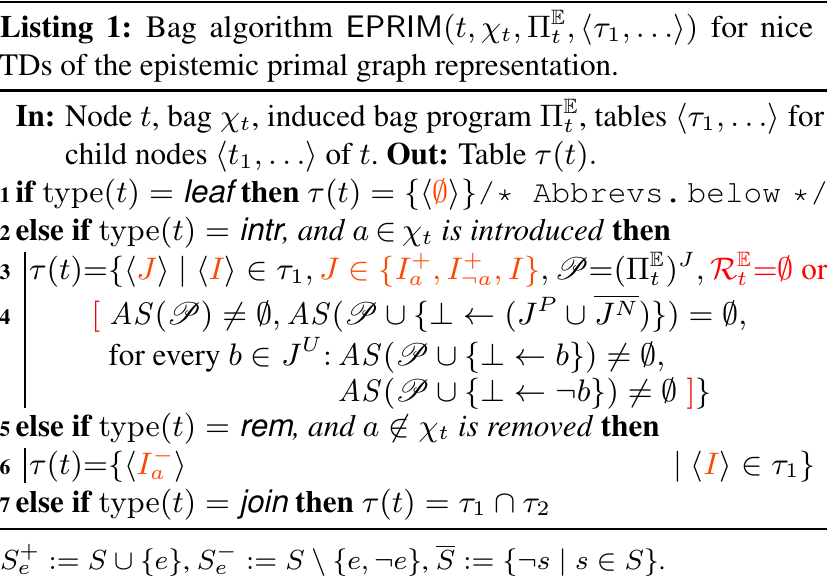}}

\begin{proposition}\label{prop:rowaspcalls}
  To compute a row in a table of \algo{EPRIM}, an ASP solver needs to be called
  at most~$2+2 \cdot |{\chi(t)}|$ times.
\end{proposition}

On an abstract level, bag algorithm~\algo{EPRIM} hence provides a method for solving epistimic
logic programs~$\Pi$ by means of plain ASP solvers based on the structure of the epistemic literals of~$\Pi$.
Whenever an epistemic atom~$a$ is removed in node~$t$, indicating that~$a$ does
not occur in any ancestor bag of~$t$, information about the ``role'' of $a$ in
any CWI is not needed anymore.
Finally, for join nodes, Line~\ref{line:join} ensures that the CWIs in~$\chi(t)$ coincide with the ones that both child bags have in common.
The last step is the evaluation of the root node. If in the root a non-empty
table is computed by \algo{EPRIM}, then the input ELP $\Pi$ has a CWV.

\begin{example}\label{ex:running3}
  The epistemic primal graph from Example~\ref{ex:running2} is a ``best
  case''-scenario for using our TD-based approach: the TD naturally separates
  the ELP into one part each for the two students, and algorithm \algo{EPRIM}
  would evaluate the two completely separately, which is exactly what intuition
  would tell us to do.  However, standard ELP solvers seem to
  struggle in this setting when the number of students increases; cf.\ e.g.\
  \cite{ijcai:BichlerMW18}.
\end{example}

From Proposition~\ref{prop:rowaspcalls}, and the facts that there are only
linearly many TD nodes in the size of the input ELP $\Pi$, and that the
number of rows per tree node is at most exponential in the treewidth of $\EP_\Pi$, we obtain the following
statement:

\begin{theorem}\label{thm:aspcalls}
  Given an input ELP $\Pi$ of size $n$, algorithm $\algo{DP}_\algo{EPRIM}$ described above makes at
  most $\bigO{2^k \cdot n}$ calls to an underlying ASP solver, where $k =
  \tw{\EP_\Pi}$.
\end{theorem}

Correctness of the algorithm presented above can be established along the same
lines as for established TD-based dynamic programming algorithms for ASP
\cite{ijcai:JaklPW09,lpnmr:FichteHMW17}. A more formal correctness argument will
be given in the next section.


%

\section{A Full Dynamic Programming Algorithm}\label{sec:dynamicprogramming}

In this section, we will extend the \algo{EPRIM} algorithm in such a way that it
no longer relies on an underlying ASP solver, but solves an ELP completely on
its own, using dynamic programming.
This new algorithm, \algo{PRIM}, will operate on the primal graph, instead of on
the epistemic primal graph, 
and
makes use of 
features of 
the entire ELP
structure.

Recall that the primal graph is defined on all atoms of an ELP, instead of just
on the ones appearing in epistemic literals.
As a result, we need to define a different solving program for TD nodes.
To this end, for the remainder of the section, assume we are a given ELP $\Pi =
(\calA, \calR)$ to solve.
Further, let $\calT = (T, \chi)$ be a nice TD of the primal graph~$\GP_\Pi$ of $\Pi$,
and $t$ a node of $\calT$.
Then, the \emph{bag rules} for~$t$, denoted $\calR_t$ are defined as the set
$\{r \mid r \in \calR, \atomof{r} \subseteq \chi(t) \}$, that is, all the rules
of~$\Pi$ that are completely ``covered'' by~$\chi(t)$.
Further, the \emph{bag program} of~$t$ is defined as $\Pi_t = (\calA \cap
\chi(t), \calR_t)$.

In order to define \algo{PRIM}, we need to define what a row of a table for a TD
node $t$ looks like.
Since \algo{PRIM}, in contrast to \algo{EPRIM}, now also needs to compute the
answer sets underlying a CWV, we start with the following, preliminary
definition.
Let $M \subseteq \chi(t)$ be an interpretation and $\calC \subseteq 2^{\chi(t)}$ a
set of interpretations w.r.t.\ $\chi(t)$. 
Then, we refer to a tuple~$\langle M, \calC \rangle$ as an \emph{answer set
tuple}. 
This construct, proposed in ~\cite{lpnmr:FichteHMW17}, directly follows the
definition of answer sets as in Definition~\ref{def:answerset}, namely, (1) set
$M$, called a \emph{witness}, is used for storing (parts of) an answer set
candidate of some ASP program, and (2) set~$\calC$, called
\emph{counterwitnesses}, holds a set of (partial) models of the GL-reduct
w.r.t.\ $M$ that are potential subsets of $M$, and hence may be counter-examples
to $M$ being extendable to an answer set.
An answer set tuple with an empty set of counterwitnesses is referred to as
\emph{proving answer set tuple}, which, vaguely speaking, proves that $M$ 
can be indeed extended to an answer set of some ASP program, 
which the tuple was constructed for.
Answer set tuples are used by algorithm~\algo{PRIM} in order to ``transport''
information---in the form of parts of models restricted to the respective bags---of
already evaluated rules of the ASP program from the leaves towards the root
during TD traversal.

With this definition in mind, we are now ready to define a row for node~$t$ used
in algorithm~\algo{PRIM}.
Such a row, called \emph{primal row}, is of the form~$\langle
\tuplecolor{\epistemiccolor}{I}, \tuplecolor{\modelcolor}{\mathcal{M}},
\tuplecolor{\killcolor}{\mathcal{K}}, \tuplecolor{\survivalcolor}{\mathcal{S}}
\rangle$, where~$I$ (printed in orange) corresponds to a CWI restricted
to~$\chi(t)$ as in \algo{EPRIM}, and sets~$\calM, \calK, \calS$ consist of
answer set tuples, where $\calM$ (printed in blue) represents a set of possible
witness answer sets for Condition~\ref{def:compat:1} of
Definition~\ref{def:compat}, $\calK$ represents possible witnesses for
disproving Conditions~\ref{def:compat:2} and~\ref{def:compat:3}, and $\calS$
represents possible witnesses for guaranteeing Condition~\ref{def:compat:4}.
In the root node~$n$ of the TD, a specific primal row~$\vec{u} \in \tab{n}$ is
required in table~$\tab{n}$ to answer the question of WV existence of $\Pi$
positively, and~\algo{PRIM} is designed to maintain primal rows accordingly.
The set~$\calM$ of answer set tuples in $\vec{u}$ is used for ensuring Condition
(1) of Definition~\ref{def:compat}, where a proving answer set tuple in~$\calM$
gives rise to an answer set of some ASP program~$\Pi^{I'}$ of some
extension~$I'\supseteq I$ of~$I$.
For ensuring Conditions (2) and (3), the set~$\calK$ in~$\vec{u}$ shall not
contain any proving program tuples, i.e., proving program tuples of~$\calK$ are
required to vanish (get ``killed'') during the TD traversal, otherwise
Conditions (2) or (3) would be violated.
Finally, $\calS$ in~$\vec{u}$ serves to establish Condition (4), where no
non-proving answer set tuple is allowed, that is, each answer set tuple needs to
``survive''.
%

\begin{definition}\label{def:proving}
  A primal row~$\langle I, \calM, \calK, \calS \rangle$ is \emph{proving} if (1)
  there is a~$\langle M, \calC \rangle \in \calM$ with~$\calC = \emptyset$, (2)
  there is no~$\langle M, \calC \rangle \in \calK$ with~$\calC = \emptyset$, and
  (3) there is no~$\langle M, \calC \rangle \in \calS$ with~$\calC \neq
  \emptyset$.
\end{definition}

{\center\includegraphics[width=.6\textwidth]{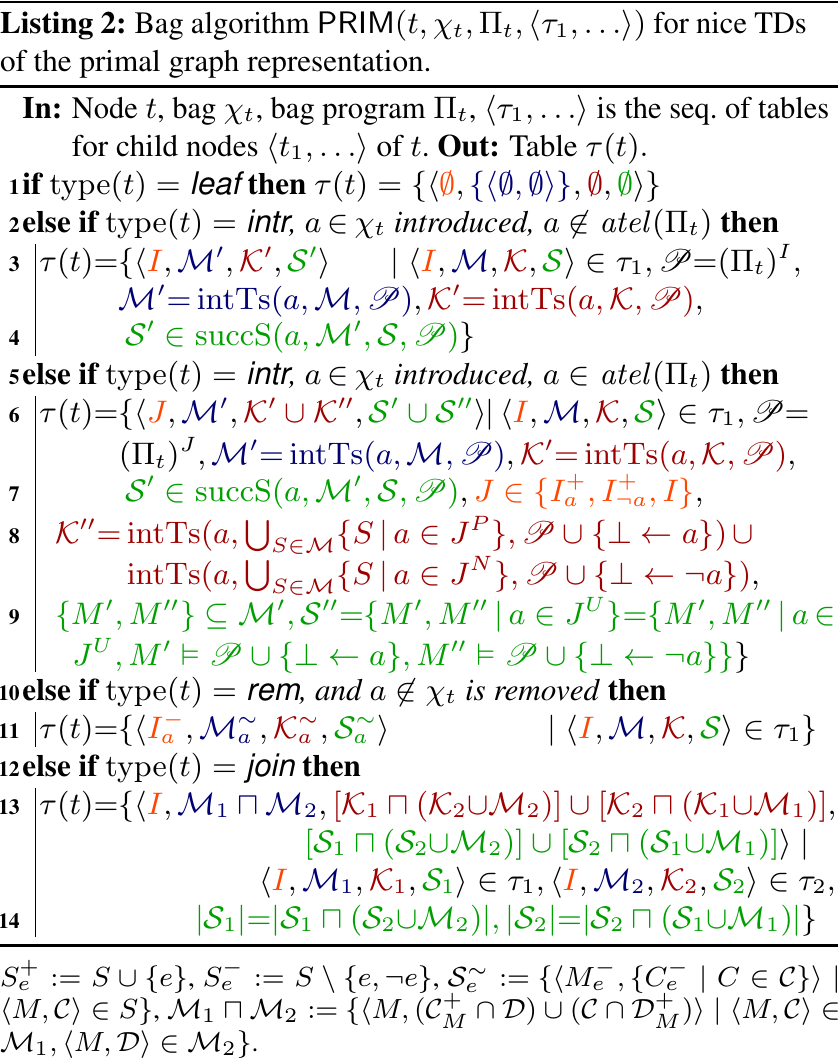}}

Algorithm~\algo{PRIM} is designed to ensure existence of such a proving primal
row in~$\tab{n}$ of root node~$n$ of the TD, iff a WV exists.
\algo{PRIM} uses the following constructs, assuming an answer set tuple~$\langle
M, \calC \rangle$, an atom~$a \in \calA$, and a program~$\scrP$.
For updating an answer set tuple, we let $\updtCand(M, \calC, \scrP) {\,=\,}
\{\langle M, \calC \cap \mods{\scrP^M} \rangle \mid M \models \scrP \}$.
When some atom~$a$ is introduced in an $\intr$-type node, we need to distinguish
between $a$ already being in the interpretation, or not.
We define $\intrCand(a, M, \calC, \scrP) = \updtCand(\MAIR{M}{a}, \bigcup_{C \in
\calC} \{ M, C, \MAIR{C}{a} \}, \scrP) \cup \updtCand(M, \calC, \scrP)$, which
is generalized to sets~$\calM$ of answer set tuples as follows:
%
%
$\updtCands(a, \calM, \scrP) {\,=\,} \bigcup_{\langle M, \calC \rangle \in \calM}
\intrCand(a, M, \calC, \scrP)$.
Finally, to obtain good runtime bounds later and at the same time still ensure
Condition (4) of Definition~\ref{def:compat} using a set~$\calS$ of answer set
tuples, we need to find, for each answer set tuple in~$\calS$, exactly one
``succeeding'' answer set tuple among the set~$\calM$ of answer set tuples.
We formalize this by defining $\survivalSets(a, \calM, \calC, \scrP) = \{ \calS'
\mid \calS' \subseteq \calM, |\calS'| = |\calS|, \text{ for every } \langle M,
\calC \rangle \in \calS : \intrCand(a, $ $M, \calC, \scrP) \cap \calS' \neq
\emptyset \}$.

Bag algorithm~\algo{PRIM}, as presented in Listing~\ref{alg:prim}, again
distinguishes between different types of tree nodes during the post-order
traversal of $\calT$.
For leafs, Line~\ref{line:leafp} returns the primal row consisting of the empty
CWI, where the second component~$\calM$ contains only the proving answer set
tuple~$\langle \emptyset, \emptyset \rangle$ (since~$\emptyset$ is the smallest
model of the empty program), and~$\calK, \calS$ are both empty as there is no
need to remove or create answer set tuples, respectively.
If an atom~$a \not\in \elitof{\Pi}$ is introduced, Line~\ref{line:IntrUpdt}
updates~$\calM$ and $\calK$.
Line~\ref{line:IntrSucc} ensures that each answer set tuple in~$\calS$ has at
least one succeeding answer set tuple in every primal row of table~$\tab{t}$.
If an atom~$a \in \elitof{\Pi}$ is introduced, the sets~$\calM$, $\calK$,
and~$\calS$ are similarly updated in Line~\ref{line:IntrUpdtE}, but the three
cases (true, false, and unknown) need to be considered when adding $a$ to $I$.
%
%
Conditions (2) and (3) are handled in Line~\ref{line:Condition23}, where answer
set tuples in~$\calM$ that violate these two conditions (for $a \in J^P$, and $a
\in J^N$, respectively) are added to~$\calK$.
For Condition (4), Line~\ref{line:Condition4} ensures that if~$a\in J^N$, there
is both a succeeding answer set tuple where~$a$ is set to true, and one
where it is false.
If an atom~$a$ is removed, $a$ is removed from the primal rows in
Line~\ref{line:removeE}, since we have processed every part of~$\Pi$ where~$a$
occurs.
Finally, for join nodes, we combine only ``compatible'' primal rows in
Line~\ref{line:joinE}.
In particular, Line~\ref{line:joinSurvive} ensures that no answer set tuple is
lost in~$\calS_1$ or~$\calS_2$ of the child primal rows.
%
%
%
%
%

In the following, we briefly mention correctness and runtime bounds for bag
algorithms~\algo{PRIM} and~\algo{EPRIM}.

\begin{theorem}[Correctness]\label{thm:primcorrectness}
  Let $\Pi$ be an ELP~$\Pi$, and $\calT = (T, \chi)$ a TD of
  $\GP_\Pi$. 
  %
  Then there is a proving row in table~$\tab{n}$ obtained by $\algo{DP}_\algo{PRIM}$ for
  root~$n$ of~$\calT$ 
  iff there is a WV
  for~$\Pi$.
\end{theorem}


\noindent Then, correctness of~$\algo{DP}_{\algo{EPRIM}}$, cf., Sec.~\ref{sec:torso}, is a special case.

\begin{corollary}[Correctness of $\algo{DP}_\algo{EPRIM}$]
  Let $\Pi$ be an ELP~$\Pi$ and $\calT = (T, \chi)$ a TD of~$\EP_\Pi$.
  Then there is a row in table $\tab{n}$ for root~$n$ obtained
  by~$\algo{DP}_\algo{EPRIM}$ iff there is a WV for~$\Pi$.
\end{corollary}

\begin{proof}[Proof (Idea)]
  %
  $\calT$ can be turned into a TD $\calT'$ of~$\GP_\Pi$ 
  by adding to each~$\chi(t)$ the set $\calA$, where~$(\Pi_t^{\EP})=(\calA,
  \cdot)$.
  $\algo{DP}_\algo{EPRIM}$ on~$\calT$ is, therefore, just a simplification of
  $\algo{DP}_\algo{PRIM}$ on~$\calT'$.
\end{proof}


\begin{theorem}[Runtime]
  $\algo{DP}_\algo{PRIM}$ runs in time~$2^{2^{2^{\calO(k)}}}\cdot|\calA|$ for
  epistemic program~$\Pi=(\calA,\calR)$, and treewidth~$k$ of~$\GP_\Pi$.
\end{theorem}


Indeed, under reasonable assumptions in computational complexity, that is, the
\emph{exponential time hypothesis (ETH)}~\cite{jcss:ImpagliazzoPZ01}, one cannot
significantly improve~$\algo{DP}_\algo{PRIM}$
since~$\algo{DP}_\algo{PRIM}$ is ETH-tight.

\begin{proposition}[cf.\ \cite{soda:FichteHP20}, Theorem~19]
  Let $\Pi = (\calA, \cdot)$ be an epistemic logic program with $\GP_\Pi$ of
  treewidth~$k$.
  Then, unless ETH fails, WV existence for~$\Pi$ cannot be decided in
  time~$2^{2^{{2^{o(k)}}}}\cdot2^{o(|\calA|)}$.
\end{proposition}

\section{Conclusions}\label{sec:conclusions}

%
%
%

This work provides the first parameterized complexity analysis of ELP solving
w.r.t.\ treewidth. Tree decompositions (TDs) have been successfully used in the
\emph{selp} ELP solver \cite{ijcai:BichlerMW18}, but for a different purpose,
namely that ELPs are rewritten into non-ground ASP programs with long rules,
which are then split up using \emph{rule decomposition} \cite{tplp:BichlerMW16}.
Our approach partitions an ELP 
according to a TD, and then solves the entire ELP by evaluating these parts in turn.
Note that this is different from 
(ELP) splitting~\cite{lpnmr:CabalarFC19,iclp:LifschitzT94}.

For future work, we aim to extend our DP algorithm to the formula evaluation
problem, which, viz.\ Theorem~\ref{thm:elpentailmentfpt}, should work in a
similar fashion to our existing algorithms, given a suitable graph
representation. Furthermore, we would like to apply our approach to other ELP
semantics; cf.\ \cite{aaai:Gelfond91,lpnmr:Gelfond11,logcom:KahlWBGZ15}. 
There, we do not anticipate large obstacles,
since most semantics are reduct-based, and the reduct is an easily exchangeable part in
our algorithms. 

\section*{Acknowledgements} M.\ Hecher and S.\ Woltran were supported by the
Austrian Science Fund (FWF) under grants Y698 and P32830.
Final version of this document will be available under AAAI proceedings.

\bibliographystyle{abbrvnat}
\bibliography{references}

\begin{thebibliography}{38}
\providecommand{\natexlab}[1]{#1}
\providecommand{\url}[1]{\texttt{#1}}
\expandafter\ifx\csname urlstyle\endcsname\relax
  \providecommand{\doi}[1]{doi: #1}\else
  \providecommand{\doi}{doi: \begingroup \urlstyle{rm}\Url}\fi

\bibitem[Bichler et~al.(2016)Bichler, Morak, and Woltran]{tplp:BichlerMW16}
M.~Bichler, M.~Morak, and S.~Woltran.
\newblock The power of non-ground rules in answer set programming.
\newblock \emph{{TPLP}}, 16\penalty0 (5-6):\penalty0 552--569, 2016.

\bibitem[Bichler et~al.(2018)Bichler, Morak, and Woltran]{ijcai:BichlerMW18}
M.~Bichler, M.~Morak, and S.~Woltran.
\newblock Single-shot epistemic logic program solving.
\newblock In \emph{Proc.\ IJCAI}, pages 1714--1720, 2018.

\bibitem[Bliem et~al.(2016)Bliem, Ordyniak, and Woltran]{ecai:BliemOW16}
B.~Bliem, S.~Ordyniak, and S.~Woltran.
\newblock Clique-width and directed width measures for answer-set programming.
\newblock In \emph{Proc. ECAI}, pages 1105--1113, 2016.
\newblock \doi{10.3233/978-1-61499-672-9-1105}.

\bibitem[Bliem et~al.(2017)Bliem, Moldovan, Morak, and
  Woltran]{ijcai:BliemMMW17}
B.~Bliem, M.~Moldovan, M.~Morak, and S.~Woltran.
\newblock The impact of treewidth on {ASP} grounding and solving.
\newblock In \emph{Proc. IJCAI}, pages 852--858, 2017.
\newblock \doi{10.24963/ijcai.2017/118}.

\bibitem[Bodlaender and Kloks(1996)]{jal:BodlaenderK96}
H.~L. Bodlaender and T.~Kloks.
\newblock Efficient and constructive algorithms for the pathwidth and treewidth
  of graphs.
\newblock \emph{J. Algorithms}, 21\penalty0 (2):\penalty0 358--402, 1996.
\newblock \doi{10.1006/jagm.1996.0049}.

\bibitem[Brewka et~al.(2011)Brewka, Eiter, and Truszczynski]{cacm:BrewkaET11}
G.~Brewka, T.~Eiter, and M.~Truszczynski.
\newblock Answer set programming at a glance.
\newblock \emph{Commun. {ACM}}, 54\penalty0 (12):\penalty0 92--103, 2011.
\newblock \doi{10.1145/2043174.2043195}.
\newblock URL \url{http://doi.acm.org/10.1145/2043174.2043195}.

\bibitem[Cabalar et~al.(2019)Cabalar, Fandinno, and {Fari{\~{n}}as del
  Cerro}]{lpnmr:CabalarFC19}
P.~Cabalar, J.~Fandinno, and L.~{Fari{\~{n}}as del Cerro}.
\newblock Splitting epistemic logic programs.
\newblock In \emph{Proc.\ LPNMR}, pages 120--133, 2019.
\newblock \doi{10.1007/978-3-030-20528-7\_10}.

\bibitem[Charwat and Woltran(2019)]{fi:CharwatW19}
G.~Charwat and S.~Woltran.
\newblock Expansion-based {QBF} solving on tree decompositions.
\newblock \emph{Fundam. Inform.}, 167\penalty0 (1-2):\penalty0 59--92, 2019.
\newblock \doi{10.3233/FI-2019-1810}.

\bibitem[Courcelle(1990)]{iandc:Courcelle90}
B.~Courcelle.
\newblock The monadic second-order logic of graphs. {I}. {R}ecognizable sets of
  finite graphs.
\newblock \emph{Inf. Comput.}, 85\penalty0 (1):\penalty0 12--75, 1990.

\bibitem[Downey and Fellows(1999)]{mcs:DowneyF99}
R.~G. Downey and M.~R. Fellows.
\newblock \emph{Parameterized Complexity}.
\newblock Monographs in Computer Science. Springer, 1999.
\newblock ISBN 978-1-4612-6798-0.
\newblock \doi{10.1007/978-1-4612-0515-9}.

\bibitem[Eiter and Gottlob(1995)]{amai:EiterG95}
T.~Eiter and G.~Gottlob.
\newblock On the computational cost of disjunctive logic programming:
  Propositional case.
\newblock \emph{Ann. Math. Artif. Intell.}, 15\penalty0 (3-4):\penalty0
  289--323, 1995.
\newblock \doi{10.1007/BF01536399}.
\newblock URL \url{http://dx.doi.org/10.1007/BF01536399}.

\bibitem[{Fari{\~{n}}as del Cerro} et~al.(2015){Fari{\~{n}}as del Cerro},
  Herzig, and Su]{ijcai:CerroHS15}
L.~{Fari{\~{n}}as del Cerro}, A.~Herzig, and E.~I. Su.
\newblock Epistemic equilibrium logic.
\newblock In \emph{Proc. IJCAI}, pages 2964--2970, 2015.

\bibitem[Fichte and Hecher(2019)]{lpnmr:FichteH19}
J.~K. Fichte and M.~Hecher.
\newblock Treewidth and counting projected answer sets.
\newblock In \emph{Proc. LPNMR}, pages 105--119, 2019.
\newblock \doi{10.1007/978-3-030-20528-7\_9}.

\bibitem[Fichte and Szeider(2015)]{ai:FichteS15}
J.~K. Fichte and S.~Szeider.
\newblock Backdoors to tractable answer set programming.
\newblock \emph{Artif. Intell.}, 220:\penalty0 64--103, 2015.
\newblock \doi{10.1016/j.artint.2014.12.001}.

\bibitem[Fichte et~al.(2017)Fichte, Hecher, Morak, and
  Woltran]{lpnmr:FichteHMW17}
J.~K. Fichte, M.~Hecher, M.~Morak, and S.~Woltran.
\newblock Answer set solving with bounded treewidth revisited.
\newblock In \emph{Proc. LPNMR}, pages 132--145, 2017.
\newblock \doi{10.1007/978-3-319-61660-5\_13}.

\bibitem[Fichte et~al.(2019{\natexlab{a}})Fichte, Hecher, and
  Pfandler]{soda:FichteHP20}
J.~K. Fichte, M.~Hecher, and A.~Pfandler.
\newblock {TE-ETH: Lower Bounds for QBFs of Bounded Treewidth}.
\newblock Preliminary version available at~\url{https://tinyurl.com/y7wnvu6w},
  2019{\natexlab{a}}.

\bibitem[Fichte et~al.(2019{\natexlab{b}})Fichte, Hecher, and
  Zisser]{cp:FichteHZ19}
J.~K. Fichte, M.~Hecher, and M.~Zisser.
\newblock An improved {GPU}-based {SAT} model counter.
\newblock In \emph{Proc. CP}, pages 491--509, 2019{\natexlab{b}}.
\newblock \doi{10.1007/978-3-030-30048-7\_29}.

\bibitem[Fichte et~al.(2019{\natexlab{c}})Fichte, Kronegger, and
  Woltran]{amai:FichteKW19}
J.~K. Fichte, M.~Kronegger, and S.~Woltran.
\newblock A multiparametric view on answer set programming.
\newblock \emph{Ann. Math. Artif. Intell.}, 86\penalty0 (1-3):\penalty0
  121--147, 2019{\natexlab{c}}.
\newblock \doi{10.1007/s10472-019-09633-x}.

\bibitem[Ganian et~al.(2017)Ganian, Ramanujan, and Szeider]{stacs:GanianRS17}
R.~Ganian, M.~S. Ramanujan, and S.~Szeider.
\newblock Combining treewidth and backdoors for {CSP}.
\newblock In \emph{Proc. STACS}, pages 36:1--36:17, 2017.
\newblock \doi{10.4230/LIPIcs.STACS.2017.36}.

\bibitem[Gelfond(1991)]{aaai:Gelfond91}
M.~Gelfond.
\newblock Strong introspection.
\newblock In \emph{Proc.\ AAAI}, pages 386--391. {AAAI} Press / The {MIT}
  Press, 1991.

\bibitem[Gelfond(2011)]{lpnmr:Gelfond11}
M.~Gelfond.
\newblock New semantics for epistemic specifications.
\newblock In \emph{Proc. LPNMR}, pages 260--265, 2011.

\bibitem[Gelfond and Lifschitz(1988)]{iclp:GelfondL88}
M.~Gelfond and V.~Lifschitz.
\newblock The stable model semantics for logic programming.
\newblock In \emph{Proc. ICLP/SLP}, pages 1070--1080. {MIT} Press, 1988.

\bibitem[Gelfond and Lifschitz(1991)]{ngc:GelfondL91}
M.~Gelfond and V.~Lifschitz.
\newblock Classical negation in logic programs and disjunctive databases.
\newblock \emph{New Generation Comput.}, 9\penalty0 (3/4):\penalty0 365--386,
  1991.

\bibitem[Gelfond and Przymusinska(1991)]{lpnmr:GelfondP91}
M.~Gelfond and H.~Przymusinska.
\newblock Definitions in epistemic specifications.
\newblock In \emph{Proc. LPNMR}, pages 245--259, 1991.

\bibitem[Gottlob et~al.(2010)Gottlob, Pichler, and Wei]{ai:GottlobPW10}
G.~Gottlob, R.~Pichler, and F.~Wei.
\newblock Bounded treewidth as a key to tractability of knowledge
  representation and reasoning.
\newblock \emph{Artif. Intell.}, 174\penalty0 (1):\penalty0 105--132, 2010.
\newblock \doi{10.1016/j.artint.2009.10.003}.

\bibitem[Impagliazzo et~al.(2001)Impagliazzo, Paturi, and
  Zane]{jcss:ImpagliazzoPZ01}
R.~Impagliazzo, R.~Paturi, and F.~Zane.
\newblock Which problems have strongly exponential complexity?
\newblock \emph{J. Comput. Syst. Sci.}, 63\penalty0 (4):\penalty0 512--530,
  2001.
\newblock \doi{10.1006/jcss.2001.1774}.

\bibitem[Jakl et~al.(2009)Jakl, Pichler, and Woltran]{ijcai:JaklPW09}
M.~Jakl, R.~Pichler, and S.~Woltran.
\newblock Answer-set programming with bounded treewidth.
\newblock In \emph{Proc. IJCAI}, pages 816--822, 2009.
\newblock URL \url{http://ijcai.org/Proceedings/09/Papers/140.pdf}.

\bibitem[Kahl et~al.(2015)Kahl, Watson, Balai, Gelfond, and
  Zhang]{logcom:KahlWBGZ15}
P.~T. Kahl, R.~Watson, E.~Balai, M.~Gelfond, and Y.~Zhang.
\newblock The language of epistemic specifications (refined) including a
  prototype solver.
\newblock \emph{J. Log. Comput.}, 25, 2015.

\bibitem[Lifschitz and Turner(1994)]{iclp:LifschitzT94}
V.~Lifschitz and H.~Turner.
\newblock Splitting a logic program.
\newblock In \emph{Proc. ICLP}, pages 23--37, 1994.

\bibitem[Lifschitz et~al.(1999)Lifschitz, Tang, and Turner]{amai:LifschitzTT99}
V.~Lifschitz, L.~R. Tang, and H.~Turner.
\newblock Nested expressions in logic programs.
\newblock \emph{Ann. Math. Artif. Intell.}, 25\penalty0 (3-4):\penalty0
  369--389, 1999.

\bibitem[Lonc and Truszczynski(2003)]{tocl:LoncT03}
Z.~Lonc and M.~Truszczynski.
\newblock Fixed-parameter complexity of semantics for logic programs.
\newblock \emph{{ACM} Trans. Comput. Log.}, 4\penalty0 (1):\penalty0 91--119,
  2003.
\newblock \doi{10.1145/601775.601779}.

\bibitem[Morak(2019)]{iclp:Morak19}
M.~Morak.
\newblock Epistemic logic programs: A different world view.
\newblock In \emph{Proc. ICLP}, pages 52--64, 2019.

\bibitem[Pearce et~al.(2009)Pearce, Tompits, and Woltran]{tplp:PearceTW09}
D.~Pearce, H.~Tompits, and S.~Woltran.
\newblock Characterising equilibrium logic and nested logic programs:
  Reductions and complexity.
\newblock \emph{{TPLP}}, 9\penalty0 (5):\penalty0 565--616, 2009.
\newblock \doi{10.1017/S147106840999010X}.
\newblock URL \url{https://doi.org/10.1017/S147106840999010X}.

\bibitem[Robertson and Seymour(1984)]{jct:RobertsonS84}
N.~Robertson and P.~D. Seymour.
\newblock Graph minors. {III.} planar tree-width.
\newblock \emph{J. Comb. Theory, Ser. {B}}, 36\penalty0 (1):\penalty0 49--64,
  1984.
\newblock \doi{10.1016/0095-8956(84)90013-3}.

\bibitem[Schaub and Woltran(2018)]{ki:SchaubW18}
T.~Schaub and S.~Woltran.
\newblock Special issue on answer set programming.
\newblock \emph{{KI}}, 32\penalty0 (2-3), 2018.
\newblock \doi{10.1007/s13218-018-0554-8}.
\newblock URL \url{https://doi.org/10.1007/s13218-018-0554-8}.

\bibitem[Shen and Eiter(2016)]{ai:ShenE16}
Y.~Shen and T.~Eiter.
\newblock Evaluating epistemic negation in answer set programming.
\newblock \emph{Artif. Intell.}, 237:\penalty0 115--135, 2016.

\bibitem[Son et~al.(2017)Son, Le, Kahl, and Leclerc]{ijcai:SonLKL17}
T.~C. Son, T.~Le, P.~T. Kahl, and A.~P. Leclerc.
\newblock On computing world views of epistemic logic programs.
\newblock In \emph{Proc. IJCAI}, pages 1269--1275, 2017.

\bibitem[Truszczynski(2011)]{birthday:Truszczynski11}
M.~Truszczynski.
\newblock Revisiting epistemic specifications.
\newblock In \emph{Logic Programming, Knowledge Representation, and
  Nonmonotonic Reasoning - Essays Dedicated to Michael Gelfond on the Occasion
  of His 65th Birthday}, 2011.

\end{thebibliography}


\end{document}